\documentclass{llncs}

\usepackage{float}

\usepackage{tikz}
\usetikzlibrary{arrows,automata}
\usetikzlibrary{positioning}
\usepackage{amssymb}
\usepackage{amsmath}
\usepackage{amsfonts}

\title{Monadic Second-Order Logic\\
with Arbitrary Monadic Predicates\thanks{A preliminary version of this work appeared in MFCS'2014~\cite{FijalkowPaperman14}.
This work was supported by the Alan Turing Institute under the EPSRC grant EP/N510129/1, 
and by the French Agence Nationale de la Recherche,
AGGREG project reference ANR-14-CE25-0017-01.}}

\author{Nathana{\"e}l Fijalkow\inst{1} \and Charles Paperman\inst{2}}
\institute{University of Warwick, United Kingdom \and Institut Math{\'e}matiques de Jussieu - Paris rive gauche}

\usepackage{tikz}
\usetikzlibrary{arrows,automata}
\usetikzlibrary{positioning}
\usepackage{amssymb}
\usepackage{amsmath}
\usepackage{amsfonts}

\newcommand{\set}[1]{\{ #1 \}}
\newcommand{\ZO}{\set{0,1}}  
\newcommand{\N}{{\mathbb N}}    

\newcommand{\F}{\mathbf{F}}  

\newcommand{\MSO}{\mathbf{MSO}}  
\newcommand{\REG}{\mathbf{REG}}  
\newcommand{\FO}{\mathbf{FO}}

\newcommand{\FOMOD}{\mathbf{FO+MOD}}

\newcommand{\BS}{\mathcal{B}\mathbf{\Sigma}}

\newcommand{\lex}{<_\mathrm{lex}}  

\newcommand{\NC}{\mathbf{NC^1}}  
\newcommand{\AC}{\mathbf{AC^0}} 
\newcommand{\ACC}{\mathbf{ACC}}  
\newcommand{\LAC}{\mathbf{LAC^0}}

\newcommand{\PP}{\mathcal{P}}

\newcommand{\NN}{\mathrm{Arb}}

\newcommand{\M}{\mathrm{Arb}_1}
\newcommand{\MU}{\mathrm{Arb}_1^{\mathrm{u}}}
\newcommand{\HDOL}{\mathrm{HD0L}}

\newcommand{\Reg}{\mathcal{R}eg}

\newcommand{\listelongue}[2]{#1^1,\ldots,#1^{#2}}
\newcommand{\listelong}[1]{#1^1,\ldots,#1^{\ell}}
\newcommand{\liste}[1]{\overline{#1}}

\newcommand{\A}{\mathcal{A}}

\begin{document}

\maketitle

\begin{abstract}
We study Monadic Second-Order Logic ($\MSO$) over finite words, extended with (non-uniform arbitrary) monadic predicates.
We show that it defines a class of languages that has algebraic, automata-theoretic and machine-independent characterizations.
We consider the \emph{regularity question}: given a language in this class, when is it regular?
To answer this, we show a \emph{substitution property} and the existence of a \emph{syntactical predicate}.

We give three applications.
The first two are to give very simple proofs that the Straubing Conjecture holds for all fragments of $\MSO$ with monadic predicates,
and that the Crane Beach Conjecture holds for $\MSO$ with monadic predicates.
The third is to show that it is decidable whether a language defined by an $\MSO$ formula with morphic predicates is regular.
\end{abstract}

\section{Introduction}
Monadic Second-Order Logic ($\MSO$) over finite words equipped with the linear ordering on positions is a well-studied and understood logic.
It provides a mathematical framework for applications in many areas such as program verification, database theory and linguistics.
In 1962, B\"uchi~\cite{Buchi62} proved the decidability of the satisfiability problem for $\MSO$ formulae.

\subsection{Uniform Monadic Predicates}
In 1966, Elgot and Rabin~\cite{ER66} considered extensions of $\MSO$ with uniform monadic predicates.
For instance, the following formula $$\forall x,\ \mathbf{a}(x) \iff x \textrm{ is prime}\enspace,$$
describes the set of finite words such that the letter $a$ appears exactly in prime positions.
The predicate ``$x$ is a prime number'' is a \emph{uniform numerical monadic} predicate.
Being \emph{numerical} means that its interpretation only depends on positions,
\textit{i.e.} $\mathbf{P} = (\mathbf{P}_n)_{n \in \N}$,
\emph{uniform} means that it can be seen as a relation over integers, 
\textit{i.e.} $\mathbf{P} \subseteq \N^k$,
and \textit{monadic} means that it has arity $1$,
\textit{i.e.} $k = 1$.
Elgot and Rabin were interested in the following question: for a uniform numerical monadic predicate $\mathbf{P} \subseteq \N$,
is the satisfiability problem of $\MSO[\le,\mathbf{P}]$ decidable?
A series of papers gave tighter conditions on $\mathbf{P}$, culminating to two final answers:
in 1984, Semenov~\cite{Semenov84} gave a characterization of the predicates $\mathbf{P}$ such that $\MSO[\le,\mathbf{P}]$ is decidable, 
and in 2006, Rabinovich and Thomas~\cite{Rabinovich07,RT06} proved this characterization to be equivalent to the predicate $\mathbf{P}$ 
being effectively profinitely ultimately periodic.
Further questions on uniform monadic predicates have been investigated.
For instance, Rabinovich~\cite{Rabinovich12} gave a solution to the Church synthesis problem for $\MSO[\le,\mathbf{P}]$, 
for a large class of predicates~$\mathbf{P}$.

In this paper, we consider the so-called numerical monadic predicates and not only the uniform ones:
such a predicate $\mathbf{P}$ is given, for each length $n \in \N$, 
by a predicate over the  first $n$ positions $\mathbf{P}_n \subseteq \set{0,\ldots,n-1}$.
The set $\M$ of these predicates contains the set $\MU$ of uniform monadic predicates.
Note that the subscript $1$ in $\M$ and $\MU$ corresponds to the arity.
A formal definition can be found in Section~\ref{defs:predicates}.

\subsection{Advice Regular Languages}
We call languages definable in $\MSO[\le,\M]$ \textit{advice regular}.
Note that no computability assumptions are made on the monadic predicates,
so this class contains undecidable languages.
Our first contribution is to give equivalent presentations of this class,
which is a Boolean algebra extending the class of regular languages:
\begin{enumerate}
	\item It has an equivalent automaton model: \textit{automata with advice}.
	\item It has an equivalent algebraic model: \textit{one-scan programs}.
	\item It has a machine-independent characterization, based on generalizations of Myhill-Nerode equivalence relations.
\end{enumerate}
This extends the equivalence between automata with advice and Myhill-Nerode equivalence relations proved in~\cite{KRSZ12}
for the special case of uniform monadic predicates.
We will rely on those characterizations to obtain several properties of the advice regular languages.
Our main goal is the following regularity question:
\begin{center} 
\emph{Given an advice regular language $L$, when is $L$ regular?}
\end{center}
To answer this question, we introduce two notions:
\begin{itemize}
	\item The \textit{substitution property}, which states that if a formula $\varphi$ together with the predicate $\mathbf{P}$ defines
	a regular language $L_{\varphi,\mathbf{P}}$, then there exists a \emph{regular predicate} $\mathbf{Q}$ such that 
	$L_{\varphi,\mathbf{Q}} = L_{\varphi,\mathbf{P}}$.
	\item The \textit{syntactical predicate} of a language $L$,
	which is the ``simplest'' predicate $\mathbf{P}_L$ such that $L \in \MSO[\le,\mathbf{P}_L]$.
\end{itemize}

Our second contribution is to show that the class of advice regular languages has the substitution property,
and that an advice regular language $L$ is regular if, and only if, $\mathbf{P}_L$ is regular.
We apply these results to the case of \emph{morphic predicates}~\cite{CT02}, and obtain the following decidability result:
given a language defined by an $\MSO$ formula with morphic predicates, one can decide whether it is regular.

\subsection{Motivations from Circuit Complexity}
Extending logics with predicates also appears in the context of circuit complexity.
Indeed, a descriptive complexity theory initiated by Immermann~\cite{Immerman87}
relates logics and circuits.
For instance, a language is recognized by a Boolean circuit of constant depth and unlimited fan-in
if, and only if, it can be described by a first-order formula with any numerical predicates of any arity,
\textit{i.e.} $\AC = \FO[\NN]$.

This correspondence led to the study of two properties, which characterize the regular languages
(Straubing Conjecture) and the languages with a neutral letter (Crane Beach Conjecture) in several fragments and extensions of $\FO[\NN]$.
The Straubing Conjecture would, if true, give a deep understanding of many complexity classes inside $\NC$. 
For instance the Straubing Conjecture for first-order logic with counting quantifiers is equivalent to the separation of $\ACC$ and $\NC$.
In the case of two-variable first-order logic, it implies tight bounds for the addition function.
Many cases of this conjecture are still open and are often equivalent to proving circuit lower bounds. 
The Crane Beach Conjecture was introduced as a model-theoretic approach to prove lower bounds, however 
this conjecture has been disproved~\cite{BILST05}.
On the positive sides, both conjectures hold in the special case of monadic predicates~\cite{BILST05,Straubing94} for several fragments.
Our third contribution is to give simple proofs of both the Straubing and the Crane Beach Conjectures for monadic predicates,
relying on our previous characterizations and extending them to abstract fragments.   
Recently, Gehrke et al~\cite{GKP16} studied first-order logic with monadic predicates but restricted to one variable,
and were able to obtain equations characterizing the regular languages in this class.

\subsection{Outline}
Section~\ref{sec:advice_regular} gives characterizations of advice regular languages,
in automata-theoretic, algebraic and machine-independent terms.
In Section~\ref{sec:substitution}, we study the regularity question,
and give two different answers: one through the substitution property,
and the other through the existence of a syntactical predicate.
The last section, Section~\ref{sec:applications}, provides applications of our results:
easy proofs that the Straubing and the Crane Beach Conjectures hold for monadic predicates
and decidability of the regularity problem for morphic regular languages.

\section{Advice Regular Languages}
\label{sec:advice_regular}
In this section, we introduce the class of advice regular languages and give several characterizations.

\subsection{Predicates}\label{defs:predicates}

A numerical predicate $\mathbf{P}$ of arity $k$ is given by $\mathbf{P} = (\mathbf{P}_n)_{n \in \N}$,
where $\mathbf{P}_n \subseteq \set{0,\ldots,n-1}^k$.
Since we mostly deal with monadic numerical predicates, we often drop the word ``monadic numerical''.
In this definition the predicates are non-uniform: for each length the predicate is different.
A predicate $\mathbf{P}$ if uniform if it is a relation over the natural numbers. More formally, if there exists $\mathbf{Q} \subseteq \N^k$ such that
for every $n$, $\mathbf{P}_n = \mathbf{Q} \cap \set{0,\ldots,n-1}^k$.
In this case we identify $\mathbf{P}$ and $\mathbf{Q}$, and see uniform predicates as subsets of~$\N^k$.

\begin{example}
	The predicate $\mathbf{first} = (\set{0})_{n \in \N}$, which is true only on the \emph{first}
	position, is uniform; we denote it by $\set{0}$.
	Similarly, the predicate $\mathbf{last} = (\set{n-1})_{n \in \N}$, which is true only for the \emph{last} position, 
	is not uniform.
\end{example}

In this paper, we will often treat predicates as words, identifying $\mathbf{P} = (\mathbf{P}_n)_{n \in \N}$ with $\mathbf{P}_n \subseteq \set{0,1}^n$.
In this case we can see $\mathbf{P}$ as a language over $\ZO$, which contains exactly one word for each length.
This simple idea is used throughout the paper, where logical formulae and automata treat predicates as words, allowing us to perform syntactical operations
on them.

We often define predicates $\mathbf{P} = (\mathbf{P}_n)_{n \in \N}$ with $\mathbf{P}_n \in A^n$ for some finite alphabet $A$.
This is not formally a predicate, but this amounts to defining one predicate for each letter in $A$,
and this abuse of notation will prove very convenient.
Similarly, any infinite word $w \in A^\omega$ can be seen as a uniform predicate.

\begin{example}
	The predicate $\mathbf{first}$ can be seen as the infinite word $1 0^\omega$, 
	and the predicate $\mathbf{last}$ as the language of finite words described by the regular expression $0^* 1$.
\end{example}

\subsection{Monadic Second-Order Logic}
The formulae we consider are monadic second-order ($\MSO$) formulae, obtained from the following grammar:
$$\varphi\ =\ \mathbf{a}(x) \mid x \le y \mid P(x) \mid \varphi \wedge \varphi \mid \neg \varphi \mid
\exists x,\ \varphi \mid \exists X,\ \varphi\enspace.$$

Here $x,y,z,\ldots$ are first-order variables, which will be interpreted by positions in the word,
and $X,Y,Z,\ldots$ are monadic second-order variables, which will interpreted by sets of positions in the word.
We say that $\mathbf{a}$ is a letter symbol, $\le$ the ordering symbol and $P,Q,\ldots$ are
the numerical monadic predicate symbols, often refered to as predicate symbols.
The notation 
$$\varphi(\listelong{P},\listelongue{x}{p},\listelongue{X}{q})$$
means that in $\varphi$,
the predicate symbols are among $\listelong{P}$,
the free first-order variables are among $\listelongue{x}{p}$
and the free second-order variables are among $\listelongue{X}{q}$.
A formula without free variables is called a sentence.
We use the notation $\liste{P}$ to abbreviate $\listelong{P}$,
and similarly for all objects (variables, predicate symbols, predicates).

We now define the semantics.
The letter symbols and the ordering symbol are always interpreted in the same way, as expected.
For the predicate symbols, the predicate symbol $P$ is interpreted by a predicate $\mathbf{P}$.
Note that $P$ is a syntactic object, while $\mathbf{P}$ is a predicate used as the interpretation of $P$.
Consider a formula $\varphi(\liste{P},\liste{x},\liste{X})$, a finite word $u$ of length $n$, 
predicates $\liste{\mathbf{P}}$ interpreting the predicate symbols from $\liste{P}$,
a valuation $\liste{\mathbf{x}}$ of the free first-order variables
and a valuation $\liste{\mathbf{X}}$ of the free second-order variables.
We define $u,\liste{\mathbf{P}},\liste{\mathbf{x}},\liste{\mathbf{X}} \models \varphi$ by induction as usual, with
\[
u,\liste{\mathbf{P}},\liste{\mathbf{x}},\liste{\mathbf{X}} \models P(y)
\quad \textrm{ if }\quad \mathbf{y} \in \mathbf{P}_n\ .
\]

A sentence $\varphi(\liste{P})$ and a tuple of predicates $\liste{\mathbf{P}}$ 
interpreting the predicate symbols from $\liste{P}$
define a language 
$$L_{\varphi,\liste{\mathbf{P}}} = \set{u \in A^* \mid u,\liste{\mathbf{P}} \models \varphi}\ .$$
Such a language is called advice regular, and the class of advice regular languages is denoted by $\MSO[\le,\M]$.

\subsection{Automata with Advice}
We introduce automata with advice.
Unlike classical automata, they have access to two more pieces of information about the word being read: 
its length and the current position.
Both the transitions and the final states can depend on those two pieces of information.
For this reason, automata with advice are (much) more expressive than classical automata, and recognize undecidable languages.
A non-deterministic automaton with advice is given by $\A = (Q,q_0,\delta,F)$
where $Q$ is a finite set of states, $q_0 \in Q$ is the initial state, 
$\delta \subseteq \N \times \N \times Q \times A \times Q$ is the transition relation
and $F \subseteq \N \times Q$ is the set of final states.
In the deterministic case, $\delta$ is a function from $\N \times \N \times Q \times A$ into $Q$.

A run of $\A$ over a finite word $u = u_0 \cdots u_{n-1} \in A^*$ is a finite word $\rho = q_0 \cdots q_n \in Q^*$
such that for all $i \in \set{0,\ldots,n-1}$, we have $(i,n,q_i,u_i,q_{i+1}) \in \delta$.
It is accepting if $(n,q_n) \in F$.
One obtains a uniform model by removing one piece of information in the transition function: the length of the word.
This automaton model is strictly weaker, and is (easily proved to be) equivalent to the one introduced in~\cite{KRSZ12},
where the automata read at the same time the input word and a fixed word called the advice.
However, our definition will be better suited for some technical aspects: 
for instance, the number of Myhill-Nerode equivalence classes exactly corresponds
to the number of states in a minimal deterministic automaton.

\begin{example}\label{ex:aut}
The language $\set{a^n b^n c^n \mid n \text{ is a prime number}}$ 
is recognized by a (deterministic) automaton with advice.
The automaton is represented in Figure~\ref{fig:aut}.
It has five states, $q_a,q_b,q_c,q_F$ and $\bot$. The initial state is $q_a$.
The transition function is defined as follows:
$$\begin{array}{llll}
\delta(i,3n,q_a,a) & = & q_a & \textrm{ if } i < n-1 \\
\delta(n-1,3n,q_a,a) & = & q_b \\
\delta(i,3n,q_b,b) & = & q_b & \textrm{ if } n \le i < 2n-1 \\
\delta(2n-1,3n,q_b,c) & = & q_c \\
\delta(i,3n,q_c,c) & = & q_c & \textrm{ if } 2n \le i < 3n-1 \\
\delta(3n-1,3n,q_c,c) & = & q_F \\
\end{array}$$
All other transitions lead to $\bot$, the sink rejecting state.
The set of final states is $F = \set{(3n,q_F) \mid n \text{ is a prime number}}$.
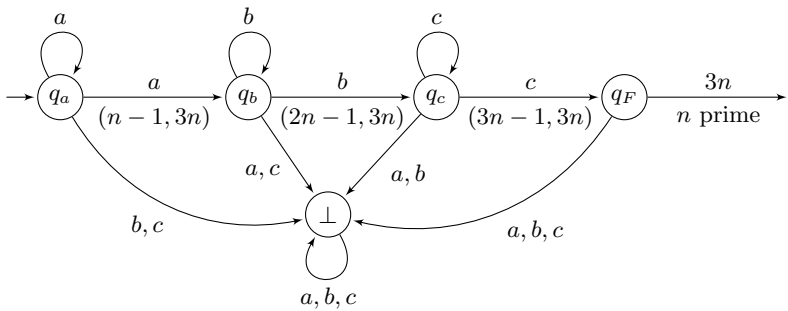
\begin{figure}
\begin{tikzpicture}[->,>=latex',shorten >=1pt,node distance=2cm,on grid,auto,initial text=,bend angle =35,
every state/.style={inner sep=0pt,minimum size=6mm},accepting/.style=accepting by arrow,loop/.style={min distance=5mm,looseness=7.5}]

{\footnotesize
\node[state,initial] (q_0)   {$q_a$};
\node[state](q_1) [right =2.5 of q_0] {$q_b$};	
\node[state](q_2) [right =2.5 of q_1] {$q_c$};
\node[state](qf) [right =2.5 of q_2] {$q_F$};
\node[state](qf2) [right =2.5 of qf, color = white] {};

\node (qab)[below right =1.5 of q_1] {};
\node[state](qb) [below =0.5 of qab] {$\bot$};
		
   \path[->]
   (q_0) edge node {$a$} node[below] {$(n-1,3n)$} (q_1)
   (q_2) edge node {$c$} node[below] {$(3n-1,3n)$} (qf)
   (q_1) edge node {$b$}  node[below] {$(2n-1,3n)$}(q_2)
   (qf) edge node {$3n$} node[below] {$n$ prime} (qf2)

   (q_2) edge[loop above, in=60,out=120 ] node {$c$} (q_2)
   (q_0) edge[loop above, in=60,out=120 ] node {$a$} (q_0)
   (q_1) edge[loop above, in=60,out=120 ] node {$b$} (q_1)
  
   (q_0) edge[bend right] node[below left] {$b,c$} (qb)
   (q_1) edge node[below left] {$a,c$} (qb)
   (q_2) edge node {$a,b$} (qb)
   (qf) edge [bend left]node {$a,b,c$} (qb)

   (qb) edge[loop below, in=240, out=300] node {${a,b,c}$} (qb);

 }
\end{tikzpicture}
\caption{The automaton for Example~\ref{ex:aut}.\label{fig:aut}}
\centering
\end{figure}
\end{example}

We mention another example, that appeared in the context of automatic structures. 
The paper~\cite{Nies07} shows that the structure $(\mathbb{Q},+)$ is automatic with advice,
which amounts to showing that the language 
$$\set{\widehat{x}\ \sharp\ \widehat{y}\ \sharp\ \widehat{z} \mid z = x + y}\enspace, $$
where $\widehat{x}$ denotes the factorial representation of the rational $x$, is advice regular.
A very difficult proof shows that this is not possible without advice~\cite{Tsankov11}.

\subsection{One-scan Programs}
Programs over monoids were introduced in the context of circuit complexity~\cite{Barrington89}:
Barrington showed that any language in $\NC$ can be computed by a program of polynomial length over a non-solvable group.
We present a simplification adapted to monadic predicates, introduced in~\cite{Straubing92} and developed in~\cite{Barrington-Straubing95}. 
We refer to~\cite[Chapter IX.4]{Straubing94} for a complete presentation.
In these works, Barrington and Straubing use Ramsey-theoretic methods to obtain non-expressibility results. 
In the remainder of this paper, we will show a generalization of these results, using a syntactic approach. 
In particular we avoid the use of Ramsey type arguments.

A one-scan program is given by $P = (M,(f_{i,n} : A \to M)_{i,n \in \N},S)$
where $M$ is a finite monoid and $S \subseteq M$.
The function $f_{i,n}$ is used to compute the effect of the $i$\textsuperscript{th}
letter of an input word of length $n$.
We say that the program $P$ accepts the word $u = u_0 \cdots u_{n-1}$ if 
$$f_{0,n}(u_0) \cdots f_{n-1,n}(u_{n-1}) \in S\enspace.$$

Note that this echoes the classical definition of recognition by monoids,
where a morphism $f : A \to M$ into a finite monoid $M$ recognizes the word 
$u = u_0 \cdots u_{n-1}$ if $f(u_0) \cdots f(u_{n-1}) \in S$.
Here, a one-scan program uses different functions $f_{i,n}$, depending on the position $i$
and the length of the word $n$.

\begin{example}
Let $U_1$ be the monoid over $\set{0,1}$ equipped with the classical multiplication. 
Consider the alphabet $A = \set{a,b}$.
We define the one-scan program $(U_1,(f_{i,n})_{i,n \in \N},\set{1})$ as follows. 
If $n$ is not a prime number, then $f_{i,n}$ is constant equal to $0$. 
Otherwise, $f_{i,n}(a) = 0$ and $f_{i,n}(b) = 1$.
Therefore, a word is accepted by this one-scan program if, and only if, its length is prime and all prime positions are labelled by the letter $b$.
\end{example}

\subsection{Myhill-Nerode Equivalence Relations}
Let $L \subseteq A^*$ and $p \in \N$, we define two equivalence relations:
\begin{itemize}
	\item $u \sim_L v$ if for all $w \in A^*$, we have $u w \in L \Longleftrightarrow v w \in L$,
	\item $u \sim_{L,p} v$ if for all $w \in A^p$, we have $u w \in L \Longleftrightarrow v w \in L$.
\end{itemize}
The relation $\sim_L$ is called the (classical) Myhill-Nerode equivalence relation,
and the second is a coarser relation, which we call the $p$-Myhill-Nerode equivalence relation.
Recall that $\sim_L$ contains finitely many equivalence classes if, and only if, $L$ is regular,
\textit{i.e.} $L \in \MSO[\le]$.

\subsection{Equivalence}
We state several characterizations of advice regular languages.

\begin{theorem}[Advice Regular Languages]\label{thm:advice_regular_languages}
Let $L$ be a language of finite words. The following properties are equivalent:
\begin{enumerate}
	\item[(1)] $L \in \MSO[\le,\M]$,
	\item[(2)] $L$ is recognized by a non-deterministic automaton with advice,
	\item[(3)] $L$ is recognized by a deterministic automaton with advice,
	\item[(4)] There exists $K \in \N$ such that for all $i,p \in \N$, 
	the restriction of $\sim_{L,p}$ to words of length $i$ contains at most $K$ equivalence classes.
	\item[(5)] $L$ is recognized by a one-scan program,
\end{enumerate}
In this case, we say that $L$ is advice regular.
\end{theorem}

This extends the Myhill-Nerode theorem proposed in~\cite{KRSZ12},
which proves the equivalence between (3) and (5) for the special case of uniform predicates.

\begin{proof}
The implication $2 \Rightarrow 3$ is proved by determinizing automata with advice,
extending the classical powerset construction.
Let $\A = (Q,q_0,\delta,F)$ be a non-deterministic automaton with advice.
We construct the deterministic automaton with advice $\A' = (Q',\set{q_0},\delta',F')$,
where $Q'$ is the powerset of $Q$, the set of final states is 
$F' = \set{(n,S) \mid \exists q \in S, (n,q) \in F}$, and the transition function $\delta'$ is defined by:
\[
\delta'(i,n,S,a) = \set{q' \in Q \mid \exists q \in S, (i,n,q,a,q') \in \delta}\ .
\]
It is easy to see that $\A$ and $\A'$ are equivalent.

The implication $3 \Rightarrow 2$ is immediate from the definitions. The
implication $1 \Rightarrow 2$ requires us to show closure properties of automata with advice 
under union, projection and complementation.
The first two closures are obtained in the exact same way as in the classical case, we do not detail them here; 
for the third case, we rely on the equivalence between $3$ and $4$,
and complement deterministic (complete) automata with advice
by simply exchanging $F$ and its complement in $\N \times Q$.

The implication $2 \Rightarrow 1$ amounts to writing a formula checking for the existence of a run.
Let $\A = (Q,q_0,\delta,F)$ be a non-deterministic automaton with advice recognizing a language $L$.

Let $\liste{X}$ be a $Q$-tuple of monadic second-order variables.
We first need to express that $\liste{\mathbf{X}}$ partitions the set of all positions of the word.
This is easily expressed by the following formula, denoted $\chi(\liste{X})$:
\[
\forall x,\ \left(\bigvee_{q \in Q} x \in X_q \right)\ \wedge\ \bigvee_{q \in Q}
\left(x \in X_q \to \bigwedge_{q' \neq q \in Q} x \notin X_{q'} \right)\ .
\]
For each $q \in Q$, we define the predicates
$\mathbf{T}^{q,a,q'}$ by $\mathbf{T}^{q,a,q'}_n = \set{i \in \N \mid \delta(i,n,q,a) = q'}$
and $\mathbf{F}^q$ by $\mathbf{F}^q = \set{n \in \N \mid (n,q) \in F}$.

The $\MSO$ formula $\varphi$ in figure~\ref{form:3} checks for the existence of an accepting run,
and uses the predicate symbols $T^{q,a,q'}$ and $F^q$.
We have $L_{\varphi,\set{\mathbf{T}^{q,a,q'},\mathbf{F}^q}} = L$.
\begin{figure*}[!t]
\normalsize
\begin{equation*}
\exists \liste{X},\
\begin{cases}
\qquad \chi(\liste{X}) \\
\wedge\ \forall x,\quad \textrm{first}(x) \Longrightarrow x \in X_{q_0} \\
\wedge\ \forall x, \forall y,\quad  y = x+1 \wedge \bigwedge_{(q,a) \in Q \times A} \bigvee_{q' \in Q} \\
\quad \qquad T^{q,a,q'}(x) \wedge x \in X_q \wedge a(x) \wedge y \in X_{q'} \\
\wedge\ \forall x,\quad \textrm{last}(x) \Longrightarrow \bigvee_{q \in Q}\quad x \in X_q \wedge F^q(x)\ .
\end{cases}
\end{equation*}
\hrulefill
\caption{Formula checking for the existence of an accepting run.\label{form:3}}
\end{figure*}

For the implication $3 \Rightarrow 4$, let $\A$ be a deterministic automaton with advice.
Let $n = i + p$, and consider the mapping $t_{i,n} : A^i \rightarrow Q$ defined for $u = u_0 \cdots u_{i-1}$ by
$$t_{i,n}(u) = \delta(i-1,n,\delta(i-2,n,\cdots \delta(0,n,q_0,u_0) \cdots, u_{i-2}), u_{i-1})\ .$$
In words, $t_{i,n}(u)$ is the state reached by $\A$ while reading $u$ of length $i$ assuming that the total word will be of length $n$.
We argue that $t_{i,n}(u) = t_{i,n}(v)$ implies $u \sim_{L,p} v$: indeed, for $w \in A^p$, after reading $u$ or $v$,
the automaton $\A$ is in the same state, so it will either accept both $uw$ and $vw$ or reject both.
Note that $t_{i,n}$ can have at most $|Q|$ different values.
Consequently, the restriction of $\sim_{L,p}$ to words of length $i$ contains at most $|Q|$ equivalence classes.

We now prove the implication $4 \Rightarrow 3$, by constructing a deterministic automaton with advice.
Its set of states is $Q = \set{0,\ldots,K-1}$.
To each word $u$ and length $n \ge 0$ we associate $\lfloor u \rfloor_n \in Q$ such that 
if $u$ and $v$ both have length $i$, then $u \sim_{L,n-i} v$ if, and only if, $\lfloor u \rfloor_n = \lfloor v \rfloor_n$.
We set $\lfloor \varepsilon \rfloor_n = 0$ for all $n$; the initial state is $0$.
The transition function is defined by $\delta(i,n,\lfloor u \rfloor_n,a) = \lfloor ua \rfloor_n$,
for \textit{some} $u$ of length $i$. 
(This is well defined: if both $u$ and $v$ have length $i$ and $\lfloor u \rfloor_n = \lfloor v \rfloor_n$, 
then $u \sim_{L,n-i} v$, so $ua \sim_{L,n-i-1} va$, 
thus $\lfloor ua \rfloor_n = \lfloor va \rfloor_n$.)
The set of final states is $$F = \set{(n,\lfloor u \rfloor_0) \mid u \in A^n \cap L}\enspace.$$
We argue that this automaton recognizes $L$:
whenever it reads $u = u_0 \cdots u_{n-1}$, the corresponding run is
$$\rho = \lfloor \varepsilon \rfloor_n \lfloor u_0 \rfloor_n \lfloor u_0 u_1 \rfloor_n \cdots 
\lfloor u_0 \cdots u_{n-1} \rfloor_n\enspace,$$
which is accepting if, and only if, $u \in L$.

The implication $3 \Rightarrow 5$ is syntactical. 
Consider a deterministic automaton with advice $\A = (Q,q_0,\delta,F)$ recognizing a language $L$.
We define $M$ to be the monoid of functions from $Q$ to $Q$, with composition as multiplication. 
Define 
$$f_{i,n}:
\begin{cases} 
	A \to M \\
	a \mapsto (q \mapsto \delta(i,n,q,a))\ ,
\end{cases}	 
$$
and $S = \set{\phi \in M \mid \phi(q_0) \in F}$.
The one-scan program $(M,(f_{i,n})_{i,n \in \N},S)$ recognizes $L$.

The converse implication $5 \Rightarrow 3$ is also syntactical.

Consider a one-scan program $(M,(f_{i,n})_{i,n \in \N},S)$ recognizing a language $L$.
Define the deterministic automaton with advice $\A = (M,1,\delta,F)$
where $1$ is the neutral element of $M$, the transition function $\delta$
is defined by $\delta(i,n,m,a) = m \cdot f_{i,n}(a)$,
and $F = \set{(n,m) \mid m \in S}$. 
The automaton $\A$ recognizes the language $L$.
\end{proof}

%

\section{The Regularity Question}
\label{sec:substitution}
In this section, we address the following question: given an advice regular language, when is it regular?
We answer this question in two different ways:
first by showing a substitution property,
and second by proving the existence of a syntactical predicate.
Note that the regularity question is not a decision problem, 
as advice regular languages are not finitely presentable,
so we can only provide non-effective characterizations 
of regular languages inside this class.

In the next section, we will show applications of these two notions:
first by proving that the Straubing property holds in this case,
and second by proving the decidability of the regularity problem 
for morphic regular languages.

\subsection{A Substitution Property}
\label{subsec:sub}
In this subsection, we prove a substitution property for $\MSO[\le,\M]$ and for $\MSO[\le,\MU]$.
We start by defining the class of (monadic) \emph{regular predicates}.

The predicates $\set{c}$ and $\mathbf{last}-c = (\set{n-1-c})_{n \in \N}$ for a given $c \in \N$ are called local predicates.
The predicates $\set{x \mid x \equiv r \bmod{q}}$ and $\mathbf{last} \equiv r \bmod{q}$ 
for given $q,r \in \N$ are called modular predicates.

\begin{theorem}[\cite{Peladeau92,Straubing94}]
\label{thm:def_reg}
Let $\mathbf{P} = (\mathbf{P}_n)_{n \in \N}$ be a predicate.
The following properties are equivalent:
\begin{enumerate}
	\item There exists a formula $\varphi(x) \in \MSO[\leq]$ over the one-letter alphabet $\set{a}$ such that 
$\mathbf{P}_n = \set{\mathbf{x} \in \set{0,n-1} \mid\ a^n,\mathbf{x} \models \varphi(x)}$.
	\item The predicate $\mathbf{P}$ is a boolean combination of local and modular predicates.
	\item The language $\mathbf{P} \subseteq A^*$ is regular.
\end{enumerate}
In this case, we say that $\mathbf{P}$ is regular.
We denote by $\Reg_1$ the class of regular (monadic) predicates.
\end{theorem}

The following theorem states the substitution property for $\MSO[\le,\M]$.

\begin{theorem}\label{thm:substitution}
For all sentences $\varphi(\liste{P})$ in $\MSO[\le,\M]$ and predicates $\liste{\mathbf{P}} \in \M$
such that $L_{\varphi,\liste{\mathbf{P}}}$ is regular,
there exist $\liste{\mathbf{Q}} \in \Reg_1$ such that 
$L_{\varphi,\liste{\mathbf{Q}}} = L_{\varphi,\liste{\mathbf{P}}}$.
\end{theorem}

The main idea of the proof is that among all predicates $\liste{\mathbf{Q}}$
such that $L_{\varphi,\liste{\mathbf{P}}} = L_{\varphi,\liste{\mathbf{Q}}}$,
there is a minimal one with respect to a lexicographic ordering,
which can be defined by an $\MSO$ formula.
The key technical point is given by the following lemma,
which can be understood as a regular choice function.

\begin{lemma}[Regular Choice Lemma]\label{lem:weak_choice}
Let $M$ be a regular language such that for all $n \in \N$, there exists a word $w \in M$ of length $n$.
Then there exists a regular language $M' \subseteq M$ such that for all $n \in \N$, there exists exactly one word $w \in M'$
of length $n$.
\end{lemma}

\begin{proof}
We equip the alphabet $A$ with a total order, inducing the lexicographic ordering $\preceq$ on $A^*$.

Let $\psi$ be an $\MSO$ formula defining $M$.
The objective is to define an $\MSO$ formula $\Psi$ such that 
$w$ satisfies $\Psi$ if, and only if, $w$ is minimal among the words of its length to satisfy $\psi$
with respect to $\preceq$.
The language defined by this formula satisfies the desired properties.

First, let $\liste{X}$ be a $A$-tuple of monadic second-order variables.
We say that $\liste{\mathbf{X}} \in A^n$ represents the word $v \in A^n$ 
if $\liste{\mathbf{X}}$ partitions the set of all positions
and for all $i \in \set{0,\ldots,n-1}, a \in A$, we have $v_i = a$ if, and only if, $i \in \mathbf{X}_a$.
The formula expressing that $\liste{\mathbf{X}}$ partitions the set of all positions is denoted by $\chi(\liste{X})$ 
(see the proof of Theorem~\ref{thm:advice_regular_languages} for the definition of this formula).

We obtain a formula $\varphi(\liste{X})$ from $\psi$ by syntactically replacing in $\psi$
each letter predicate $\mathbf{a}(x)$ by $x \in X_a$.
The following property holds:
for all $\liste{\mathbf{X}} \in A^*$, if $\liste{\mathbf{X}}$ represents the word $v$,
then $\liste{\mathbf{X}} \models \varphi(\liste{X})$ is equivalent to $v \models \psi$.

Now, we define a formula $\theta(\liste{X})$ such that for all words $w \in A^*$
and $\liste{\mathbf{X}} \in A^*$, if $\liste{\mathbf{X}}$ represents a word $v$, 
then $w, \liste{\mathbf{X}} \models \theta(\liste{X})$ if, and only if, 
$w \preceq v$.
There are two cases: either $w = v$, or $w \prec v$, so the formula $\theta(\liste{X})$ is a disjunction of two formulae,
the first stating that $\liste{\mathbf{X}}$ represents $w$:
$$\forall x,\ \bigwedge_{a \in A} \left( x \in X_a \iff \mathbf{a}(x)\right)\ ,$$
and the second stating that $w \prec v$, where $v$ is represented by $\liste{\mathbf{X}}$:
$$\exists x, \left(\bigvee_{a < b \in A} \mathbf{a}(x) \wedge x \in X_b \right) \wedge 
\left(\forall y,\ y < x \to \left(\bigvee_{a \in A} \mathbf{a}(y) \wedge y \in X_a \right) \right)\ .$$

The $\MSO$ formula $\Psi$ that selects the minimal word in $M$ is given by:
$$\psi\ \wedge\ \forall \liste{X},\ \left(\chi(\liste{X}) \wedge \varphi(\liste{X})\right) \to \theta(\liste{X})\ .$$
This concludes the proof.
\end{proof}

We now prove Theorem~\ref{thm:substitution} relying on Lemma~\ref{lem:weak_choice}.

\begin{proof}
Consider $\varphi(\liste{P})$ a sentence and $\liste{\mathbf{P}}$ predicates
such that $L_{\varphi,\liste{\mathbf{P}}}$ is regular. We write $L$ for $L_{\varphi,\liste{\mathbf{P}}}$.
Let $\theta$ be an $\MSO$ formula defining $L$.

Consider the language $M = \set{\liste{\mathbf{X}} \in (\ZO^\ell)^* \mid L_{\varphi,\liste{\mathbf{X}}} = L}$.

We first argue that $M$ is regular.

As in the proof of Lemma~\ref{lem:weak_choice}, we introduce $\liste{Y}$ a $A$-tuple of monadic second-order variables,
used to represent words in $A^*$.
We obtain a formula $\psi(\liste{Y})$ from $\theta$ by syntactically replacing in $\theta$
each letter predicate $\mathbf{a}(x)$ by $x \in X_a$.
The following property holds:
for all $\liste{\mathbf{Y}} \in A^*$, if $\liste{\mathbf{Y}}$ represents the word $v$,
then $\liste{\mathbf{Y}} \models \psi(\liste{Y})$ is equivalent to $v \models \theta$.

Consider the following formula in $\MSO[\le]$ over the alphabet $\ZO^\ell$:
\[
\forall \liste{Y},\ \chi(\liste{Y}) \implies \left( \varphi(\liste{X}) \Longleftrightarrow \psi(\liste{Y}) \right)\ ,
\]
it describes the language $M$: the word $\liste{\mathbf{X}} \in (\ZO^\ell)^n$ satisfies this formula if
for any word $v \in A^n$ represented by $\liste{\mathbf{Y}}$, $v$ is in $L$ if, and only if, $v,\liste{\mathbf{X}} \models \varphi(\liste{X})$.

\medskip
Now, we note that for all $n \in \N$, there exists a word in $M$ of length $n$, namely $\liste{\mathbf{P}}_n$.
Thus Lemma~\ref{lem:weak_choice} applies, so there exists $M' \subseteq M$ a regular language
so that for all $n \in \N$, there exists a unique word in $M'$ of length $n$,
which we denote by $\liste{\mathbf{Q}}_n$.
Thanks to Theorem~\ref{thm:def_reg}, this yields a tuple of regular predicates $\liste{\mathbf{Q}}$ such that 
$L_{\varphi,\liste{\mathbf{Q}}} = L_{\varphi,\liste{\mathbf{P}}}$.
\end{proof}

We proved the substitution property for $\MSO[\le,\M]$.
We now prove that it also holds for $\MSO[\le,\MU]$; note that this is not implied by the previous case.
We first prove it over infinite words, and then transfer the result to finite words.

\begin{theorem}\label{thm:sub_infinite_words}
For all sentences $\varphi(\liste{P})$ in $\MSO[\le,\MU]$ and predicates $\liste{\mathbf{P}} \in \MU$
such that $L_{\varphi,\liste{\mathbf{P}}}$ is $\omega$-regular,
there exist regular predicates $\liste{\mathbf{Q}} \in \MU$ such that 
$L_{\varphi,\liste{\mathbf{P}}} = L_{\varphi,\liste{\mathbf{Q}}}$.
\end{theorem}


\begin{proof}
Consider $\varphi(\liste{P})$ a sentence in $\MSO[\le,\MU]$
and $\liste{\mathbf{P}}$ predicates such that $L_{\varphi,\liste{\mathbf{P}}}$ is $\omega$-regular,
denote it $L$.

Consider the following language:
$$M = \set{\liste{\mathbf{X}} \in (\set{0,1}^\ell)^\omega 
\mid \ L_{\varphi,\liste{\mathbf{X}}} = L}\ . $$
Relying on the same arguments as in the proof of Theorem~\ref{thm:substitution},
we show that $M$ is $\omega$-regular.
It is also non-empty since it contains $\liste{\mathbf{P}}$.
It follows from B{\"u}chi's Theorem that it contains a ultimately periodic word $\liste{\mathbf{Q}}$.
Seen as predicates, $\liste{\mathbf{Q}}$ are regular monadic uniform predicates, 
and $L_{\varphi,\liste{\mathbf{P}}} = L_{\varphi,\liste{\mathbf{Q}}}$,
which concludes the proof.
\end{proof}

The following theorem states the substitution property for $\MSO[\le,\MU]$.

\begin{theorem}
\label{thm:sub_uniform}
For all sentences $\varphi(\liste{P})$ in $\MSO[\le,\MU]$ and predicates $\liste{\mathbf{P}} \in \MU$
such that $L_{\varphi,\liste{\mathbf{P}}}$ is regular,
there exist regular predicates $\liste{\mathbf{Q}} \in \MU$ such that 
$L_{\varphi,\liste{\mathbf{P}}} = L_{\varphi,\liste{\mathbf{Q}}}$.
\end{theorem}

\begin{proof}
We consider $\varphi(\liste{P})$ a sentence in $\MSO[\le,\MU]$ 
and predicates $\liste{\mathbf{P}}$ such that $L_{\varphi,\liste{\mathbf{P}}}$ is regular.
Let $\flat$ be a fresh letter (not in $A$), we denote by $A_\flat$ the alphabet $A \cup \set{\flat}$.
We explain how to transform $\varphi(\liste{P})$ into a formula $\widehat{\varphi}(\liste{P})$ 
in $\MSO[\le,\MU]$ over the alphabet $A_\flat$ satisfying:
for all $u$ in $A^*$ and predicates $\liste{\mathbf{Q}}$, we have:
\begin{equation}\label{eq:finite_infinite}
u,\liste{\mathbf{Q}} \models \varphi \ \Longleftrightarrow\ 
u \cdot \flat^\omega,\liste{\mathbf{Q}} \models \widehat{\varphi}
\end{equation}
$\widehat{\varphi}(\liste{P})$ is obtained from $\varphi(\liste{P})$ by guarding every first-order quantifiers:
the subformula $\exists y, \theta(y)$ is turned into the subformula
$\exists y, \neg \flat(y) \wedge \widehat{\theta}(y)$.
The equivalence~(\ref{eq:finite_infinite}) is easily proved by induction.
Consider the formula $\psi(\liste{P})$ defined by 
$$\widehat{\varphi}(\liste{P})\ \wedge\ \exists y,\quad \left(\forall x \geq y,\ \flat(x)\ \wedge\ \forall x < y,\ \neg \flat(x)\right)$$
and the predicates $\liste{\mathbf{P}}$,
they define the language $L_{\varphi,\liste{\mathbf{P}}} \cdot \flat^\omega$
thanks to the equivalence~(\ref{eq:finite_infinite}),
so it is $\omega$-regular.

From Theorem~\ref{thm:sub_infinite_words},
we get a tuple of regular predicates $\liste{\mathbf{Q}} \in \MU$ such that 
$L_{\psi,\liste{\mathbf{Q}}} = L_{\varphi,\liste{\mathbf{P}}} \cdot \flat^\omega$.
It follows that $L_{\varphi,\liste{\mathbf{P}}} = L_{\varphi,\liste{\mathbf{Q}}}$, which concludes the proof.
\end{proof}

We note that the substitution property does not hold over binary predicates. 
In fact, one can show much worse: given $M$ a deterministic Turing machine, 
one can construct a universal formula $\varphi_M(\liste{P})$ (\textit{i.e.} with only universal quantifiers)
with binary predicates such that $L_{\varphi_M,\liste{\mathbf{P}}} = a^*$ if and only if $\liste{\mathbf{P}}$ represents the run of $M$.
In other words, even if the language of the formula is rather simple, it can use its predicates to perform arbitrarily complicated computations.

\subsection{The Syntactical Predicate}
\label{subsec:syntactical_predicate}
In this subsection, we define the notion of syntactical predicate for an advice regular language.
The word ``syntactical'' here should be understood in the following sense:
the syntactical predicate $\mathbf{P}_L$ of $L$ is the most regular predicate that describes the language $L$.
In particular, we will prove that $L$ is regular if, and only if, $\mathbf{P}_L$ is regular.

Let $L$ be an advice regular language.
We define the predicate $\mathbf{P}_L = (\mathbf{P}_{L,n})_{n \in \N}$.
Thanks to Theorem~\ref{thm:advice_regular_languages},
there exists $K \in \N$ such that for all $i,p \in \N$, the restriction
of $\sim_{L,p}$ to words of length $i$ contains at most $K$ equivalence classes.
Denote $Q = \set{0,\ldots,K-1}$ and $\Sigma = (Q \times A \to Q) \uplus Q$, 
where $Q \times A \to Q$ is the set of functions from $Q \times A$ to $Q$.
We define $\mathbf{P}_{L,n} \in \Sigma^n$.

Let $i,n \in \N$.
Among all words of length $i$, we denote by $u^{i,n}_1,u^{i,n}_2,\ldots$ the lexicographically minimal representatives 
of the equivalence classes of $\sim_{L,n-i}$, enumerated in the lexicographic order:
\begin{equation}\label{eq:lexicographic_order}
u^{i,n}_1 \lex u^{i,n}_2 \lex u^{i,n}_3 \lex \cdots
\end{equation}
In other words, $u^{i,n}_q$ is minimal with respect to the lexicographic order $\lex$ among all words of length $i$
in its equivalence class for $\sim_{L,n-i}$.
Thanks to Theorem~\ref{thm:advice_regular_languages}, there are at most $K$ such words for each $i,n \in \N$.

We define $\mathbf{P}_{L,n}(i)$ (the $i$\textsuperscript{th} letter of $\mathbf{P}_{L,n}$) by:

\begin{equation}\label{eq:p_l1}
\mathbf{P}_{L,n}(i)(q,a) = q'\ \textrm{ if }\ u^{i,n}_q \cdot a \sim_{L,n-i-1} u^{i+1,n}_{q'}\ , \textrm{ for } i < n - 1
\end{equation}
\begin{equation}\label{eq:p_l2}
\mathbf{P}_{L,n}(n-1)(q)\ \textrm{ if }\ u^{n,n}_q \in L\ .
\end{equation}

Intuitively, the predicate $\mathbf{P}_{L}$ describes the transition function with respect to the equivalence relations $\sim_{L,p}$.
We now give an example.

\begin{example}
\begin{figure}
\centering
\begin{tikzpicture}[every node/.style={circle,inner sep=0pt,minimum size = .8cm}]
    \fill[gray!25] (0,0) ellipse (0.75 and 1.5);
    \fill[gray!25] (2.5,0) ellipse (0.75 and 1.5);
	\fill[gray!25] (5,0) ellipse (0.75 and 1.5);
 	\fill[gray!25] (7.5,0) ellipse (0.75 and 1.5);

	\node at (0,2) {$0$};
	\node at (2.5,2) {$1$};
	\node at (5,2) {$2$};
	\node at (7.5,2) {$3$};


	\node[fill=gray!60] (a) at (0,1) {$a$};
	\node[fill=gray!60] (b) at (0,-1) {$b$};
	\node[fill=gray!60] (aa) at (2.5,1) {$aa$};
	\node[fill=gray!60] (ab) at (2.5,0) {$ab$};
	\node[fill=gray!60] (ba) at (2.5,-1) {$ba$};

	\node[fill=gray!60] (aaa) at (5,1) {$a^3$};
	\node[fill=gray!60] (aba) at (5,0) {$aba$};
	\node[fill=gray!60] (bab) at (5,-1) {$bab$};

	\node[fill=gray!60] (aq) at (7.5,1) {$a^4$};
	\node[fill=gray!100] (abq) at (7.5,-1) {$(ab)^2$};


  \draw (a) edge [bend left = 35]  node[fill = white,minimum size = 0.5cm] {$a$} (aa);
  \draw (a) edge [bend left = 25]  node[fill = white,minimum size = 0.5cm] {$b$} (ab);
  \draw (b) edge [bend left = -10]  node[fill = white,minimum size = 0.5cm] {$b$} (aa);
  \draw (b) edge [bend left = -35]  node[fill = white,minimum size = 0.5cm] {$a$} (ba);

  \draw (aa) edge [bend left = 35]  node[fill = white,minimum size = 0.5cm] {$a,b$} (aaa);
  \draw (ab) edge [bend left = 25]  node[fill = white,minimum size = 0.5cm] {$a$} (aba);
  \draw (ab) edge [bend left = 25]  node[fill = white,minimum size = 0.5cm] {$b$} (aaa);
  \draw (ba) edge [bend left = -10]  node[fill = white,minimum size = 0.5cm] {$a$} (aaa);
  \draw (ba) edge [bend left = -35]  node[fill = white,minimum size = 0.5cm] {$b$} (bab);

  \draw (aaa) edge [bend left = 35]  node[fill = white,minimum size = 0.5cm] {$a,b$} (aq);
  \draw (aba) edge [bend left = 25]  node[fill = white,minimum size = 0.5cm] {$a$} (aq);
  \draw (aba) edge [bend left = -10]  node[fill = white,minimum size = 0.5cm] {$b$} (abq);
  \draw (bab) edge [bend left = 5]  node[fill = white,minimum size = 0.5cm] {$a,b$} (aq);

\end{tikzpicture}
\caption{The predicate $\mathbf{P}_L$ (here $\mathbf{P}_{L,4}$) for $L = (ab)^* + (ba)^* b$.\label{fig:syntactical_predicate}}
\end{figure}
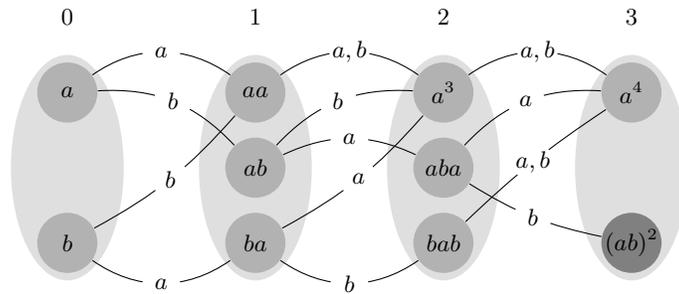
Consider the language $L = (ab)^* + (ba)^* b$. We represent $\mathbf{P}_{L,4}$ in figure~\ref{fig:syntactical_predicate}.
Each circle represents an equivalence class with respect to $\sim_{L,4}$, inside words of a given length.
For instance, there are three equivalence classes for words of length $3$: 
$a^3, aba$ and $bab$. 
Note that these three words are the minimal representatives of their equivalence classes with respect to the lexicographic order.
For the last position (here $3$), the equivalence class of $(ab)^2$ (which is actually reduced to $(ab)^2$ itself) is darker
since it belongs to the language $L$.
\end{example}

We state the main property of the predicate $\mathbf{P}_L$.

\begin{theorem}
\label{thm:syntactical_predicate}
Let $L$ be an advice regular language.
Then $L$ is regular if, and only if, $\mathbf{P}_L$ is regular.
\end{theorem}

The proof is split in two lemmas, giving each direction. 
We start by the \emph{if} direction.

\begin{lemma}\label{lem:p_l}
Let $L$ be an advice regular language.
Then $L \in \MSO[\le,\mathbf{P}_L]$.
\end{lemma}

\begin{proof}
From the definition of $\mathbf{P}_L$, it is easy to see that
the word $u$ of length $n$ belongs to $L$ if, and only if, there exists $X : \set{0,\ldots,n} \to Q$
such that:
$$
\begin{array}{lll}
& \forall q \in Q,\ & X(0) = q \iff \mathbf{P}_{L,n}(0)(0,u_0) = q \\
\wedge\ & \forall q,q' \in Q, \forall i < n-1,\ & X(i+1) = q'\ \Longleftarrow\ X(i) = q \wedge \mathbf{P}_{L,n}(i)(q,u_i) = q' \\
\wedge\ & \forall q \in Q,\ & X(n) = q \implies \mathbf{P}_{L,n}(n-1)(q) \ .
\end{array} $$
This can be written down as an $\MSO$ formula with the predicate $\mathbf{P}_L$.
\end{proof}

The \emph{if} direction of Theorem~\ref{thm:syntactical_predicate} follows from Lemma~\ref{lem:p_l},
because if $\mathbf{P}_L$ is regular, then $L \in \MSO[\le,\mathbf{P}_L] = \MSO[\le]$,
so $L$ is regular.

For the \emph{only if} direction, we prove a stronger statement that 
will be useful in Section~\ref{section:morphic}. Informally, we prove 
that the syntactic predicate $\mathbf{P}_L$ 
is the \emph{least} predicate required to define in $\MSO$ the language $L$.

\begin{lemma}\label{lem:p_l_converse}
Let $L$ be an advice regular language defined with the predicates $\liste{\mathbf{P}}$.
Then $\mathbf{P}_L \in \MSO[\le,\liste{\mathbf{P}}]$.
\end{lemma}

\begin{proof}
Assume that $L$ is defined by an $\MSO$ formula $\theta$ with the predicates $\liste{\mathbf{P}}$.
Then the three equations~\eqref{eq:lexicographic_order},~\eqref{eq:p_l1} and~\eqref{eq:p_l2} defining $\mathbf{P}_L$
can be written down as an $\MSO$ formula with the predicates $\liste{\mathbf{P}}$.

To this end, we represent words as monadic second-order variables as in the proof of Lemma~\ref{lem:weak_choice}.
A $A$-tuple $\liste{X}$ of monadic second-order variables represents the word $v \in A^n$ 
if $\liste{\mathbf{X}}$ partitions the set of all positions up to position $n$,
and for all $i \in \set{0,\ldots,n-1}$, we have $v_i = a$ if, and only if, $i \in \mathbf{X}_a$.

Denote by $\chi(\liste{X},x)$ the $\MSO$ formula expressing that
$\liste{\mathbf{X}}$ partitions the set of all positions up to position $\mathbf{x}$.
Similarly, denote by $\Xi(\liste{X},x)$ the $\MSO$ formula expressing that
$\liste{\mathbf{X}}$ partitions the set of all positions from the position $\mathbf{x} + 1$.

The formulae for~\eqref{eq:lexicographic_order} and~\eqref{eq:p_l2} 
make use of the formulae $\chi(\liste{X},x)$ and $\theta$.
We omit them as they are easy to write down, and focus on~\eqref{eq:p_l1}.

The first step is to construct a formula $\varphi(\liste{X},\liste{Y},w)$ such that
if $w$ has length $n$, $\liste{\mathbf{X}}$ represents $u$ and $\liste{\mathbf{Y}}$ represents $v$ both of length $i$,
then $w,\liste{\mathbf{X}},\liste{\mathbf{Y}},i \models \varphi(\liste{X},\liste{Y},x)$
if, and only if, $u \sim_{L,n-i} v$.
Define $\varphi(\liste{X},\liste{Y},x)$ as:
$$\forall \liste{Z},\ \Xi(\liste{Z},x) \implies (\phi(\liste{X},x,\liste{Z}) \iff \phi(\liste{Y},x,\liste{Z}))\ ,$$
where the formula $\phi(\liste{X},x,\liste{Z})$ is obtained from $\theta$ by syntactically replacing in $\theta$
each letter predicate $\mathbf{a}(y)$ by $(y \le x \wedge y \in X_a) \vee (y > x \wedge y \in Z_a)$.

The second step is to construct a finite number of formulae $\gamma_\ell(\liste{X},x)$ for $q \in Q$ such that 
if $w$ has length $n$ and $\liste{\mathbf{X}}$ represents $u$ of length $i$, then 
$w,\liste{\mathbf{X}},i \models \gamma_\ell(\liste{X},x)$
if, and only if, there are exactly $q - 1$ words of length $i$ that are 
(i) pairwise not equivalent with respect to $\sim_{L,n-i}$, 
(ii) not equivalent to $u$ with respect to $\sim_{L,n-i}$,
and (iii) smaller than $u$ with respect to the lexicographic order.

We can now put the pieces together and give a formula for~\eqref{eq:p_l1}:
\begin{multline*}
\forall x, \forall \liste{X}, \forall \liste{Y},\  
\bigwedge_{a \in A,\ q,q' \in Q} x \in X_a \wedge \chi(\liste{X},x) \wedge \chi(\liste{Y},x) \implies \\
P_{L,q,a,q'}(z) \iff (\gamma_q(\liste{X},x) \wedge \gamma_{q'}(\liste{Y},x))\ .
\end{multline*}
It follows that $\mathbf{P}_L$ is definable in $\MSO[\le,\liste{\mathbf{P}}]$.
\end{proof}

\section{Applications}
\label{sec:applications}

In this section we show several consequences of Theorem~\ref{thm:advice_regular_languages}
(characterization of the advice regular languages), Theorem~\ref{thm:substitution}
(a substitution property for advice regular languages)
and Theorem~\ref{thm:syntactical_predicate} (a syntactical predicate for advice regular predicates).

The first two applications are about two conjectures, the Straubing and the Crane Beach Conjectures, 
introduced in the context of circuit complexity.
We first explain the motivations for these two conjectures,
and show very simple proofs of both of them in the special case of monadic predicates.


The third application shows that one can determine, given an $\MSO$ formula
with morphic predicates, whether it defines a regular language.

\subsection{A Descriptive Complexity for Circuit Complexity Classes}
\label{subsec:circuits}
We first quickly define some circuit complexity classes.
The most important here is $\AC$, the class of languages defined by boolean circuits of bounded depth and polynomial size,
and its subclass $\LAC$ where the circuits have linear size.
From $\AC$, adding the modular gates gives rise to $\ACC$.
Finally, the class of languages defined by boolean circuits of logarithmic depth, polynomial size and fan-in $2$ is denoted by $\NC$.
Separating $\ACC$ from $\NC$ remains a long-standing open problem.

One approach to better understand these classes is through descriptive complexity theory, 
giving a perfect correspondence between circuit complexity classes and logical formalisms.
Unlike what we did so far, the logical formalisms involved in this descriptive complexity
theory use predicates of any arity (we focused on predicates of arity one).
A $k$-ary predicate $\mathbf{P}$ is given by $(\mathbf{P}_n)_{n \in \N}$,
where $\mathbf{P}_n \subseteq \set{0,\ldots,n-1}^k$.
We denote by $\NN$ the class of all predicates,
and by $\Reg$ the class of regular predicates as defined in~\cite{Straubing94}.

We recall the notations for some of the classical classes of formulae: $\FO$ (first-order quantifiers),
$\FOMOD$ (first-order and modular quantifiers: $\exists^{r,q} x, \varphi(x)$ reads 
``the number of $x$ satisfying $\varphi(x)$ is equal to $r \bmod{q}$''),
$\FO^2$ (first-order with at most two variables)
and $\BS_k$ (at most $k-1$ alternations of $\exists$ and $\forall$ quantifiers).

\begin{theorem}[\cite{Immerman87,BCST92,GL84,KLPT06}]\hfill
\begin{enumerate}
	\item[(1)] $\AC  = \FO[\NN]$,
	\item[(2)] $\LAC = \FO^2[\NN]$,
	\item[(3)] $\ACC = (\FOMOD)[\NN]$.
\end{enumerate}
\end{theorem}  

Two conjectures have been formulated on the logical side,
which aim at clarifying the relations between different circuit complexity classes.
They have been stated and studied in special cases, we extrapolate them here to all fragments.
We first need to give an abstract notion of (logical) fragment.
Several such notions can be found in the bibliography, with more or less strong syntactic restrictions (see~\cite{KL12}). 
In this paper, we use a minimalist definition of fragment: we only require to be allowed to substitute predicates within each formula. 
Remark that this property is not restrictive and is satisfied by all classical fragments of $\MSO$.
We fix the universal signature, containing infinitely many predicate symbols for each arity.
Let $\F$ be a class of formulae over this signature and $\PP$ a class of predicates,
describing the fragment $\F[\PP]$ by:
$$\F[\PP] = \set{L_{\varphi,\liste{\mathbf{P}}} \mid \varphi \in \F \wedge \liste{\mathbf{P}} \in \PP}\ .$$

The first property, called the Straubing property, characterizes the regular languages (denoted by $\REG$) inside
a larger fragment.

\begin{definition}[Straubing Property]
$\F[\PP]$ has the Straubing property if:
all regular languages definable in $\F[\PP]$ are also definable in $\F[\PP \cap \Reg]$.

In symbols,
\[
\F[\PP] \cap \REG = \F[\PP \cap \Reg]\ .
\]
\end{definition}

This statement appears for the first time in~\cite{BCST92},
where it is proved that $\FO[\NN]$ has the Straubing property,
relying on lower bounds for $\AC$ and an algebraic characterization of $\FO[\Reg]$.
Following this result, Straubing conjectures in~\cite{Straubing94} that 
$(\FOMOD)[\NN]$ and $\BS_k[\NN]$ have the Straubing property for $k \ge 1$.
Recently, this conjecture has been extended to $\FO^2[\NN]$ (see~\cite{KLPT06}).
If true, it would imply the separation of $\ACC$ from $\NC$,
and for the $\FO^2$ case, tight lower bounds on the addition of two integers in binary. 

We already mentioned that several several fragments have the Straubing property,
as for instance, $\mathbf{\Sigma_1}[\NN]$, $\FO[\le,\M]$ and $(\FOMOD)[\le,\M]$, 
as proved by Straubing and Barrington~\cite{Straubing92,Barrington-Straubing95}
by using Ramsey arguments for one scan programs and algebraic characterizations of these fragments.
In this paper, we give a simpler syntactical proof that all fragments $\F[\le,\M]$ have the Straubing property.
The second property, called the Crane Beach property, characterizes the languages having a neutral letter, 
and is derived from a conjecture proposed by Th{\'e}rien for the special case of first-order logic
and finally disproved in the article~\cite{BILST05}.

\begin{definition}[Neutral letter]
A language $L$ has a neutral letter $e \in A$ if for all words $u,v$, we have 
$uv \in L$ if, and only if, $uev \in L$.
\end{definition}

\begin{definition}[Crane Beach Property]
$\F[\PP]$ has the Crane Beach property if:
all languages having a neutral letter definable in $\F[\PP]$ are definable in $\F[\le]$.
\end{definition}

Unfortunately, as mentioned, the Crane Beach property does not hold in general.
\begin{theorem}[\cite{BILST05,SchweikardtPHD}]
There exists a non-regular language having a neutral letter definable in $\FO[\NN]$.
\end{theorem}
A deeper understanding of the Crane Beach property specialized to first-order logic can be found in~\cite{BILST05}.
In particular, it has been shown that $\FO[\le,\M]$ has the Crane Beach property.
In this paper, we give a simple proof that $\MSO[\le,\M]$ has the Crane Beach Property.

\subsection{The Straubing Conjecture for Advice Regular Languages}
\label{subsec:straubing}

\begin{theorem}
\label{thm:straubing}
All fragments $\F[\le,\M]$ have the Straubing property. 
\end{theorem}

This is actually a straightforward corollary of Theorem~\ref{thm:substitution}.
\begin{proof}
Let $\varphi \in \F$ such that $L_{\varphi,\liste{\mathbf{P}}}$ with $\liste{\mathbf{P}} \in \M$ is regular.
Thanks to Theorem~\ref{thm:substitution}, there exist $\liste{\mathbf{Q}} \in \Reg_1$ such that 
$L_{\varphi,\liste{\mathbf{Q}}} = L_{\varphi,\liste{\mathbf{P}}}$.
This concludes the proof.
\end{proof}

We state a corollary of Theorem~\ref{thm:straubing}.

\begin{corollary}
For all $k \ge 1$, $\BS_k[\le,\M]$ has the Straubing property.
\end{corollary}



We conclude this subsection by remarking that the substitution property does not hold over infinite words, even for monadic predicates.
This follows from the simple observation that adding the ``bit-predicate'' to first-order logic allows us to express all of monadic second-order logic.
Formally, the bit-predicate $\mathbf{B}$ is defined by $\mathbf{B}(x,y)$ holds if the $y$\textsuperscript{th} bit of the binary representation of $x$ is $1$.
Roughly speaking, in the setting of infinite words the bit-predicate can make use of the infinite number of positions to talk about any finite set, 
hence first-order logic with the bit-predicate expresses all of weak monadic second-order logic, which coincides with monadic second-order logic.

Now the Straubing Property over infinite words for first-order logic reads:
$$\FO[\NN] \cap \omega \REG = \FO[\Reg]\ .$$
This would imply $\MSO[\le] \subseteq \FO[\Reg]$, which does not hold: the parity language,
defined by $L = \set{u \cdot \natural^\omega \mid u \in \set{a,b}^* \textrm{ has an even number of } a}$
belongs to $\MSO[\le]$, but not to $\FO[\Reg]$~\cite{STT95}.

\subsection{The Crane Beach Conjecture for Advice Regular Languages}
\label{subsec:crane_beach}
In this subsection, we show that the Crane Beach Conjecture holds for advice regular languages.
\begin{theorem}
\label{thm:crane_beach}
$\MSO[\le,\M]$ has the Crane Beach property. 
\end{theorem}

The proof is a simple corollary of Theorem~\ref{thm:advice_regular_languages}.

\begin{proof}
Recall that a language over finite words $L$ has a neutral letter $e \in A$
if for all words $u$ and $v$, we have $uv \in L$ if, and only if, $uev \in L$.
In other words, $u \sim_L ue$.

\vskip1em
Let $L$ be an advice regular language, thanks to Theorem~\ref{thm:advice_regular_languages},
there exists $K \in \N$ such that for all $i,p \in \N$, 
the restriction of $\sim_{L,p}$ to words of length $i$ contains at most $K$ equivalence classes.
	
We argue that $\sim_L$ contains at most $K$ equivalence classes (without both restrictions to words of a given length).
Indeed, assume to the contrary that there are $K+1$ words that are pairwise non-equivalent with respect to~$\sim_L$.
By iterating the equivalence $u \sim_L ue$, we obtain $K+1$ words of the same length (the maximal length of the $K+1$ original words),
which are still pairwise non-equivalent with respect to $\sim_L$.
For two non-equivalent words $u,v$, there exist a third word $w$ witnessing the non-equivalence: 
$uw \in L$ but $vw \notin L$ or the other way around.
Again by padding with the neutral letter $e$, we obtain non-equivalence witnesses for each pair of the $K+1$ words
of the same length (the maximal length of the $\frac{(K+1)(K+2)}{2}$ witnesses).
Hence we have $K+1$ words of the same length which are not equivalent with respect to $\sim_{L,p}$ for the same $p$,
a contradiction.
It follows that $L$ is regular, \textit{i.e} $L \in \MSO[\le]$.
\end{proof}

\subsection{Morphic Regular Languages}\label{section:morphic}
\label{subsec:morphic}
In this subsection, we apply Theorem~\ref{thm:syntactical_predicate} to the case of morphic predicates,
and obtain the following result: given an $\MSO$ formula with morphic predicates,
it is decidable whether it defines a regular language.

The class of morphic predicates was first introduced by Thue in the context of combinatorics on words,
giving rise to the HD0L systems.
Formally, let $A,B$ be two finite alphabets, 
$\sigma : A^* \to A^*$ a morphism, $a \in A$ a letter such that $\sigma(a) = a \cdot u$ for some $u \in A^+$
and $\varphi : A^* \to B^*$ a morphism. 
This defines the sequence of words $\varphi(a),\varphi(\sigma(a)),\varphi(\sigma^2(a)),\ldots$,
which converges to a finite or infinite word.
An infinite word obtained in this way is said to be morphic.

We see morphic words as predicates, and denote by $\HDOL$ the class of morphic predicates.
We call the languages definable in $\MSO[\le,\HDOL]$ morphic regular.

\begin{theorem}
\label{thm:morphic}
The following problem is decidable:
given $L$ a morphic regular language, is $L$ regular?
Furthermore, if $L$ is regular, then we can construct a finite automaton for $L$.
\end{theorem}

The proof of this theorem goes in two steps:
\begin{itemize}
	\item first, we reduce the regularity problem for a morphic regular language $L$
to deciding the ultimate periodicity of $\mathbf{P}_L$,
	\item second, we show that $\mathbf{P}_L$ is morphic.
\end{itemize}
Hence we rely on the following result: given a morphic word, it is decidable whether it is ultimately periodic.
The decidability of this problem was conjectured $30$ years ago and proved recently and simultaneously by Durand and Mitrofanov~\cite{Durand13,Mitrofanov12}. 

The first step is a direct application of Theorem~\ref{thm:syntactical_predicate}.
For the second step, observe that thanks to Lemma~\ref{lem:p_l_converse},
we have $\mathbf{P}_L \in \MSO[\le,\HDOL]$.
We conclude with the following lemma, which follows from the characterization of morphic words
as being those automatically presentable with the lexicographic ordering~\cite{RM02}.

\begin{lemma}
$\HDOL$ is closed under $\MSO$-interpretations,
\textit{i.e.}
if $\mathbf{P}$ is an infinite word such that $\mathbf{P} \in \MSO[\le,\HDOL]$, then $\mathbf{P} \in \HDOL$.
\end{lemma}

Furthermore, all constructions in this proof are effective, and if $\mathbf{P}_L$ is ultimately periodic,
then one can compute the threshold and the period,
and derive from them a finite automaton for $L$.

As a corollary, we also obtain from Theorem~\ref{thm:morphic} the decidability of $\MSO[\le,\HDOL]$.
Indeed, from a language in $\MSO[\le,\HDOL]$, we first determine whether it is regular, and:
if it is regular, then determine whether it is empty by looking at the (effectively constructed) finite automaton recognizing it, 
and if it is not regular, then it is non-empty (since the empty language is regular).
We stress however that this result can be obtained with a much more direct proof~\cite{CT02}.


\section*{Acknowledgments}
We thank Thomas Colcombet and Sam van Gool for fruitful discussions,
and Jean-{\'E}ric Pin for his advice.
We are grateful to the anonymous referees for their \textbf{very} constructive comments.

\bibliographystyle{alpha}
\bibliography{bib}
\end{document}